\documentclass[letterpaper, 10 pt, conference]{ieeeconf}  % Comment this line out if you need a4paper

\IEEEoverridecommandlockouts                              % This command is only needed if 
\usepackage{cite}
\usepackage{amsmath,amssymb,amsfonts}
\usepackage{algorithmic}
\usepackage{graphicx}
\usepackage{textcomp}
\usepackage{algorithm}
\usepackage{subfig}
\usepackage[colorlinks=true]{hyperref}

\newtheorem{theorem}{Theorem}[section]
\newtheorem{lemma}[theorem]{Lemma}

\newtheorem{proposition}[theorem]{Proposition}

\newtheorem{problem}[theorem]{Problem}

    \title{\LARGE \bf
Remote  Estimation for Markov Jump Linear Systems: \\ A Distributionally Robust Approach
}

\author{Ioannis Tzortzis, Themistoklis Charalambous   and Charalambos D. Charalambous
\thanks{The authors are with the Department
of Electrical and Computer Engineering, University of Cyprus, Nicosia, Cyprus. E-mails: (\{tzortzis.ioannis,\ charalambous.themistoklis,\ chadcha\}@ucy.ac.cy).  T. Charalambous is also with the Department of Electrical Engineering and Automation, School of Electrical Engineering, Aalto University, Espoo, 02150, Finland. }
\thanks{This work has been partly funded by MINERVA, a European Research Council (ERC) project funded under the European Union's Horizon 2022 research and innovation programme (Grant agreement No. 101044629).}
}

\begin{document}
% \author{Ioannis Tzortzis, Themistoklis Charalambous and Charalambos D. Charalambous
% \thanks{The authors are with  the Department of Electrical and Computer Engineering, University of Cyprus, Nicosia 1678, Cyprus (emails: \{tzortzis.ioannis,\ charalambous.themistoklis,chadcha\}@ucy.ac.cy). T. Charalambous is also with the Department of Electrical Engineering and Automation, School of Electrical Engineering, Aalto University, Espoo, 02150, Finland.}}

\maketitle

\begin{abstract}
This paper considers the problem of remote state estimation for Markov jump linear systems in the presence of uncertainty in the posterior mode probabilities. Such uncertainty may arise when the estimator receives noisy or incomplete measurements over an unreliable communication network. To address this challenge, the estimation problem is formulated within a distributionally robust framework, where the true posterior is assumed to lie within a total variation distance  ball centered at the nominal posterior. The resulting minimax formulation yields an estimator that extends the classical MMSE solution with additional terms that account for mode uncertainty. A tractable implementation is developed using a distributionally robust variant of the first-order generalized pseudo-Bayesian  algorithm. A
numerical example is provided to illustrate the applicability and effectiveness of the  approach.
\end{abstract}

% \begin{IEEEkeywords}
% Bayesian estimation, distributionally robust, Markov jump linear systems.
% \end{IEEEkeywords}

\section{Introduction}
\label{sec.Intro}
Markov jump linear systems (MJLS) model dynamical systems that switch between multiple modes, each governed by its own set of linear dynamics.  These systems have been widely studied due to their broad range of practical applications, particularly in state estimation \cite{logothetis1999expectation}. 
The goal is to jointly estimate the continuous-valued state and discrete mode sequence from noisy observations. Applications span areas such as target tracking \cite{mcginnity2000multiple,magnant2016bayesian,tzortzis2021linear}, signal processing \cite{liu2017state}, telecommunication \cite{yin2016distributed}, and fault detection \cite{zhao2016detection,hibey1999conditional}.
While computing the exact minimum mean-squared error (MMSE) estimator is intractable due to the exponential growth of mode hypotheses, suboptimal algorithms such as the generalized pseudo-Bayesian (GPB) and interacting multiple model (IMM) filters provide practical alternatives \cite{ahn2019bayesian,blom1988interacting,elliott2005state}.   These classical estimators typically assume perfect knowledge of the system’s underlying probability distributions.
In practice, however, estimation often occurs in networked control systems, where communication between sensors and estimators is subject to noise, delay, packet drops, or partial channel state information \cite{Tzortzis:2025,schenato2007foundations}. These limitations can introduce uncertainty in the posterior mode probabilities, even when the system model is  known. To ensure reliable estimation performance, distributionally robust methods are desirable.

Distributionally robust  approaches provide a modelling framework for optimal state estimation of MJLS by assuming that the true underlying probability distribution is unknown  and lies within an ambiguity set of probability distributions. 
Recent work has focused on robustness with respect to uncertain MC transition probabilities \cite{wang2023distributionally,jilkov2004online,doucet2002recursive,
li2016h,orguner2006online,zhang2009mode}, which is useful when the source of uncertainty is known. However, when no such prior information is available, ambiguity can instead be modeled directly at the posterior level, leading to alternative robust estimation methods.

Motivated by the above discussion, we consider the problem of distributionally robust state and mode estimation for MJLS,  with uncertainty placed on the estimator's posterior mode probabilities. Such uncertainty may arise due to
(a) limited or noisy observations, which affect the output/emission model;
(b) ambiguity in the prior mode probabilities, stemming from parametric uncertainty in the MC; or
(c) a combination of both.
This uncertainty affects the posterior  and, in turn,  the optimal state estimate.
To address this, we adopt a distributionally robust formulation in which the true posterior  lies within an ambiguity set defined by the total variation distance (TVD) metric.  The estimation problem is then formulated as a minimax optimization, where the estimator  minimizes the mean-squared error (MSE) against  the worst-case posterior distribution in the ambiguity set. The resulting estimator extends the classical MMSE estimator by incorporating  correction terms that depend on the level of uncertainty. 
For implementation, we introduce a tractable  distributionally robust variant 
of the GPB algorithm.

The remainder of this paper is organized as follows. Section \ref{sec:prob.form} formulates the distributionally robust state estimation problem within a Bayesian framework. Section \ref{sec:prob.sol} presents the optimal solution, including  a distributionally robust variant of the GPB algorithm. Section \ref{sec:sim.results} illustrates the proposed approach through a numerical example. Section \ref{sec:conclusions} concludes the paper.

\section{Problem Formulation}
\label{sec:prob.form}
\subsection{State and observation processes dynamics}
Consider the following discrete-time MJLS
\begin{subequations}\label{MJLS}
\begin{align}
x_{k}&=A(\theta_{k})x_{k-1}+B(\theta_{k})w_{k},\label{MJLS.a}\\
y_k&=C(\theta_k)x_k+D(\theta_k)v_k,\label{MJLS.b}
\end{align}
\end{subequations}
where $k\in \mathbb{N}:=\{1,2,\dots\}$ denotes the time index, $x_k\in  \mathbb{R}^{n_x}$ is the state vector process, and $y_k\in  \mathbb{R}^{n_y}$ is the observation vector process. The jump process $\{\theta_k, k\in \mathbb{N}\}$ is generated from a discrete time, homogeneous MC taking values in a finite space $\Theta=\{1,2,\dots,n_\theta\}$ with the $(i,j)-$th element of the transition probability matrix $\Pi$ equal to
\begin{equation}\label{MC}
\pi_{ij}:=P(\theta_{k}=j|\theta_{k-1}=i),\ \forall i,j\in \Theta.
\end{equation} 
The initial state $x_0\sim{\cal N}(x_0;\bar{x}_0,X_0)$ is normally distributed with mean $\bar{x}_0$ and covariance matrix $X_0$, and the noises $w_k\sim {\cal N}(w_k;0,W_k)$ and  $v_k\sim {\cal N}(v_k;0,V_k)$ are independent, identically distributed Gaussian processes. 
For each $\theta_k\in \Theta$, the real-valued matrices $A(\theta_k)\in \mathbb{R}^{n_x\times n_x}$, $B(\theta_k)\in \mathbb{R}^{n_x\times n_w}$, $C(\theta_k)\in \mathbb{R}^{n_y\times n_x}$ and $D(\theta_k)\in \mathbb{R}^{n_y\times n_v}$ are assumed to have bounded and measurable entries.
It is also assumed that $x_0$, $\{\theta_k\}$, $\{w_k\}$ and $\{v_k\}$ are mutually independent for all $k$.

The system is monitored by a remote estimator, which receives the observation sequence $y_{1:k}:=\{y_1,\dots,y_k\}$  over an unreliable  network. Based on these measurements, the estimator must recursively estimate both the hidden state $x_k$ and the mode $\theta_k$, neither of which is directly observable. Although the dynamics are linear within each mode, the dependence on the hidden mode sequence $\{\theta_k\}$ makes the overall estimation problem nonlinear.

% The MJLS  \eqref{MJLS} is a hybrid dynamical system comprising  both a continuous state-space and a finite set of discrete modes evolving according to a MC. The estimation objective  is to infer the system's hybrid state from noisy observations. To this end, we adopt a Bayesian estimation framework.

\subsection{Bayesian estimation}
Bayesian estimation for MJLS consists of two steps, namely, the \textit{prediction step} for evaluating the joint prior conditional distribution
\begin{equation}\label{prior}
p(x_k,\theta_k|y_{1:k-1})=p(x_k|\theta_k,y_{1:k-1})p(\theta_k|y_{1:k-1}),
\end{equation}
and the \textit{correction step} for updating the joint posterior conditional distribution
\begin{equation}\label{posterior}
p(x_k,\theta_k|y_{1:k})=p(x_k|\theta_k,y_{1:k})p(\theta_k|y_{1:k}).
\end{equation}
% \IT{Both steps are executed by the remote estimator, using the observation history and the system model.}
% In  \eqref{prior} and \eqref{posterior}, $y_{1:t}:=\{y_1,\dots,y_t\}$ denotes the observation history up to time  $t\in \mathbb{N}$. 
At each  step $k-1$, we assume that the mode probabilities $p(\theta_{k-1}|y_{1:k-1})$ and the  posterior probability density functions (pdfs) $p(x_{k-1}|\theta_{k-1},y_{1:k-1})$ are available to the estimator from the previous filtering step.  
By applying the prediction and  correction steps,  one obtains both the mode probabilities  $p(\theta_k|y_{1:k})$ and the posterior pdfs  $p(x_k|\theta_k,y_{1:k})$.

\subsubsection{Prediction step} 
Using the Chapman-Kolmogorov equation, the prior conditional pdf of $x_k$ given $\theta_k$ and past  observations $y_{1:k-1}$ is
\begin{multline}\label{apriori_pdf}
p(x_k|\theta_k,y_{1:k-1})\\
=\int  p(x_k|x_{k-1}, \theta_{k})p(x_{k-1}|\theta_{k},y_{1:k-1})dx_{k-1},
\end{multline}
where  $p(x_k|x_{k-1}, \theta_{k})$ follows from the system dynamics. The term  $p(x_{k-1}|\theta_{k},y_{1:k-1})$ is evaluated by marginalizing over $\theta_{k-1}$ using the backward Bayes rule, while the prior mode probabilities are given by
% \begin{multline}
% p(x_{k-1}|\theta_k,y_{1:k-1})\\ =\sum_{\theta_{k-1}}p(x_{k-1}|\theta_{k-1},y_{1:k-1})p(\theta_{k-1}|\theta_k,y_{1:k-1}),
% \end{multline}
% with
% \begin{equation*}
% p(\theta_{k-1}|\theta_k,y_{1:k-1})=\frac{p(\theta_k|\theta_{k-1})p(\theta_{k-1}|y_{1:k-1})}{\sum_{\theta_{k-1}}p(\theta_k|\theta_{k-1})p(\theta_{k-1}|y_{1:k-1})}.
% \end{equation*}
% The prior mode probabilities $p(\theta_k|y_{1:k-1})$ are given by 
\begin{equation}\label{prior.mode.prob}
p(\theta_k|y_{1:k-1})=\sum_{\theta_{k-1}}p(\theta_k|\theta_{k-1})p(\theta_{k-1}|y_{1:k-1}),
\end{equation}
where $p(\theta_k|\theta_{k-1})$ is specified by \eqref{MC}.

\subsubsection{Correction step}
The joint posterior conditional distribution \eqref{posterior} is specified by 
\begin{align}
p(x_k|\theta_k,y_{1:k})& \propto p(y_k|x_k,\theta_k)p(x_k|\theta_k,y_{1:k-1}),\label{mode.prob.aa}\\
p(\theta_k|y_{1:k})& \propto p(y_k|\theta_k,y_{1:k-1})p(\theta_k|y_{1:k-1}),\label{mode.prob}
\end{align}
where 
\begin{equation*}
p(y_k|\theta_k,y_{1:k-1})=\int p(y_k|x_k,\theta_k)p(x_k|\theta_k,y_{1:k-1})dx_k,
\end{equation*}
and $p(y_k|x_k,\theta_k)$ is specified by \eqref{MJLS.b}. Notice that, both \eqref{mode.prob.aa} and \eqref{mode.prob} are unnormalized expressions of the posterior pdfs.

\subsection{Optimality criterion and ambiguity set}
In practice, the nominal posterior mode probabilities $\mu_k(\theta_k):=p(\theta_k|y_{1:k})$ are computed by the remote estimator using the received measurements. However, when communication is unreliable due to packet drops, delays, sensor noise, or partial channel state information, the estimator's information may be incomplete or degraded. As a result, the computed posterior $\mu_k$ may deviate from the true posterior distribution over modes, which in turn affects the quality of the resulting state estimates. To account for such deviations, we adopt a distributionally robust formulation that models uncertainty in the posterior mode probabilities.

 Let $\varphi_k(y_{1:k})$ denote a function that maps the observation history $y_{1:k}$ to an estimate of the state $x_k$.  As an optimality criterion, we consider  the MSE cost 
\begin{align}\label{mse.cost}
V_k(\varphi,\mu)&=\mathbb{E}\big[||x_k-{\varphi}_k(y_{1:k})||^2\big]\nonumber\\
&:=\int ||x_k-{\varphi}_k(y_{1:k})||^2p(x_k|y_{1:k})dx_k,
 \end{align}
 where the marginal pdf $p(x_k|y_{1:k})$ is given by
 \begin{equation}\label{mse.prob}
p(x_k|y_{1:k})=\sum_{\theta_k}p(x_k|\theta_k,y_{1:k})\mu_k(\theta_k).
 \end{equation}
Substituting \eqref{mse.prob} into \eqref{mse.cost} yields the equivalent form
 \begin{equation}\label{nominal.cost}
V_k(\varphi,\mu)=\sum_{\theta_k}\mu_k(\theta_k)     \mathbb{E}\big[||x_k-{\varphi}_k(y_{1:k})||^2|\theta_k,y_{1:k}  \big].
\end{equation}
 
The  MMSE estimator  is $\hat{\varphi}_k(y_{1:k})=\mathbb{E}[x_k|y_{1:k}]$.
However, from \eqref{nominal.cost}, it is clear that the optimal estimate depends on the posterior mode probabilities $\mu_k(\theta_k)$ which may be uncertain, as discussed earlier. To reduce the influence of uncertainty on the performance of the optimal estimator, we quantify any uncertainty regarding the posterior mode probabilities $\mu_k(\theta_k)$ (hereafter referred to as the \emph{nominal} posterior mode probabilities) by a collection of posterior mode probabilities $\nu_k(\theta_k):=\bar{p}(\theta_k|y_{1:k})$ (hereafter referred to as the  \emph{true} posterior mode probabilities) lying within an ambiguity set defined by the TVD metric, as described below.

% The MMSE estimator is $\hat{\varphi}_k(y_{1:k})=\mathbb{E}[x_k|y_{1:k}]$. Since it depends on the possibly inaccurate posterior mode probabilities $\mu_k(\theta_k)$ (hereafter referred to as the \emph{nominal} posterior mode probabilities), we model this uncertainty by allowing the posterior mode probabilities $\nu_k(\theta_k):=\bar{p}(\theta_k|y_{1:k})$ (hereafter referred to as the  \emph{true} posterior mode probabilities) to lie within an ambiguity set defined by the TVD metric, as described below.

\emph{Ambiguity set:} Let $\mathbb{P}(\Theta)$ denote the set of all probability distributions over the finite mode set $\Theta$.
% \begin{equation*}
% \mathbb{P}(\Theta):=\big\{p=(p_1,\dots,p_{n_\theta})\in \mathbb{R}_+^{n_\theta}: \sum_{i=1}^{n_\theta} p_i=1\big\}.
% \end{equation*}
The ambiguity set consists of all possible true posterior mode probabilities $\nu_k\in \mathbb{P}(\Theta)$ that lie within a ball of radius $R_{TV}(k)\in [0,1]$, centered at the nominal distribution $\mu_k$, and defined using the TVD metric as
\begin{multline}\label{ambiguity.class}
\mathbb{B}_k(\mu_k):=\big\{\nu_k\in \mathbb{P}(\Theta):\\ \frac{1}{2}\sum_{i\in \Theta}|\nu_k(\theta_k=i)-\mu_k(\theta_k=i)|\leq R_{TV}(k)\big\}.
\end{multline}

\emph{Distributionally robust criterion:} The distributionally robust counterpart of the MSE cost \eqref{mse.cost} is given by
\begin{equation}\label{mse.DRO.cost}
V_k(\varphi,\nu)=\sum_{\theta_k}\nu_k(\theta_k)     \mathbb{E}\big[||x_k-{\varphi}_k(y_{1:k})||^2|\theta_k,y_{1:k}  \big],
\end{equation}
where the expectation is taken with respect to $p(x_k|\theta_k,y_{1:k})$.

\subsection{DRO Problem}
% \IT{This leads to the following distributionally robust estimation problem, in which the estimator seeks to minimize the worst-case MSE over all mode distributions within the ambiguity set.} 
\begin{problem}
Find an optimal estimator $\hat{\varphi}_k(y_{1:k})\in \Phi$ and the worst-case posterior mode probabilities $\nu^*_k(\theta_k)\in \mathbb{B}_k(\mu_k)$ that solve the minimax problem
\begin{equation}\label{DRO1}
\min_{\varphi_k\in \Phi}\max_{\nu_k\in \mathbb{B}_k(\mu_k)} V_k(\varphi,\nu),
\end{equation}
where $V_k(\varphi,\nu)$ is defined in \eqref{mse.DRO.cost}, and  $\Phi$ denotes the set of all admissible estimators of $x_k$.
\end{problem}

% A key feature of the DRO problem is that $R_{TV}$ controls the level of conservatism. When $R_{TV} = 0$, the ambiguity set collapses to the nominal distribution, and the DRO problem reduces to the classical one. As $R_{TV}$ increases, the estimator accounts for increasingly uncertain scenarios.

\section{Optimal Solution}
\label{sec:prob.sol}
% This section  presents the optimal solution to the DRO problem \eqref{DRO1}. We first solve the inner maximization problem, and then derive the solution to the outer minimization.

\subsection{Solution to the inner optimization problem}
\subsubsection{Maximization of a linear functional} 
The distributionally robust criterion \eqref{mse.DRO.cost} can be rewritten as 
% \begin{equation}\label{equiv.dro.cost}
% V_k(\varphi,\nu)=\sum_{\theta_k} \nu_k(\theta_k)L_k^\varphi (\theta_k),
% \end{equation}
% where
% \begin{equation}\label{expected.term}
% L_k^\varphi (\theta_k)=\int ||x_k-{\varphi}_k(y_{1:k})||^2p(x_k|\theta_k,y_{1:k})dx_k.
% \end{equation}
\begin{align}\label{expected.term}
    V_k(\varphi,\nu)&=\sum_{\theta_k} \nu_k(\theta_k)\int ||x_k{-}{\varphi}_k(y_{1:k})||^2p(x_k|\theta_k,y_{1:k})dx_k\nonumber\\
    &=:\sum_{\theta_k} \nu_k(\theta_k)L_k^\varphi (\theta_k).
\end{align}
The superscript in $L_k^\varphi(\cdot)$ indicates its dependence  on the estimator $\varphi_k\in \Phi$. 
%For all $\theta_{k-1}=i\in \Theta$, let $L_\varphi(i)\in \mathbb{R}^{1\times n_\theta}$  defined by
%\begin{multline}\label{ell_function}
%L^\varphi(i):=\Big[\ell^\varphi(\theta_k=1,\theta_{k-1}=i),\ell^\varphi(\theta_k=2,\theta_{k-1}=i),\\ \dots, \ell^\varphi(\theta_k=n_\theta,\theta_{k-1}=i)\Big].
%\end{multline}
%Then, \eqref{equiv.dro.cost} can be written as follows
%\begin{subequations}
%\begin{align}
%&V_k(\varphi,\Pi)=\sum_{i\in \Theta}V_{k,i}(\varphi,\pi)\label{MSE.total}\\
%&V_{k,i}(\varphi,\pi)=L^\varphi(i)\pi_i^T,\quad i=1,2,\dots,n_\theta.\label{MSE.vf}
%\end{align}
%\end{subequations}
%Notice that, for $i\neq j$, $i,j\in \Theta$, $V_{k,i}(\cdot)$ and $V_{k,j}(\cdot)$ are decoupled, and hence, their optimization can be addressed independently. In particular, the inner optimization in \eqref{DRO1} becomes
%\begin{align}\label{DRO2}
%&\max_{\sr{\pi_i\in \mathbb{B}_R(i)}{\forall i=1,2,\dots,n_\theta}} V_k(\varphi,\Pi) \ \equiv \
%\sum_{i\in \Theta}\Big(\max_{\pi_i\in \mathbb{B}_R(i)} V_{k,i}(\varphi,\pi)\Big).
%\end{align}
%This decoupling allows for parallel computation of $\pi_i^*$ over all $i \in \Theta$,  making the right-hand side of \eqref{DRO2} easier to be solved.
%Next, we address the problem of maximizing a linear functional subject to the ambiguity class, as given by \eqref{ambiguity.class}.
We arrange the values of $L_k^\varphi(\theta_k)$ for all modes $\theta_k\in \Theta$ in a descending order. Let
\begin{equation*}
L_{\max}^\varphi(k):=\max_{\theta_k\in \Theta}L_k^\varphi(\theta_k),\quad
L_{\min}^\varphi(k):=\min_{\theta_k\in \Theta}L_k^\varphi(\theta_k),
\end{equation*}
and define $\Theta^0(k)$ and $\Theta_0(k)$ as the sets of modes attaining $L_{\max}^\varphi(k)$ and $L_{\min}^\varphi(k)$, respectively. The remaining modes are grouped into sets $\Theta_j(k)$, $j=1,2,\dots,r$, according to this descending order, with
\begin{equation}
L^\varphi_{\Theta_j}(k):=\min_{\theta_k\in \Theta\setminus \{\Theta^0(k)\cup (\cup_{s=1}^j \Theta_{s-1}(k))\}}L_k^\varphi(\theta_k).
\end{equation}

In the next theorem,   $(d)^+:=\max\{0,d\}$.
\begin{theorem}\label{main.theorem}
The optimal solution to the inner optimization problem in \eqref{DRO1} is given by
\begin{align}\label{payoff}
V_{k}(\varphi,\nu^*)&=L^\varphi_{\max}(k)\nu^*_k(\Theta^0(k)) + L^\varphi_{\min}(k)\nu^*_k(\Theta_0(k)) \nonumber\\
&\quad + \sum_{s=1}^r L^\varphi_{\Theta_s}(k)\nu^*_k(\Theta_s(k)).    
\end{align}
where the maximizing mode probabilities $\nu_k^*(\cdot)\in \mathbb{B}_k(\mu_k)$ are  given by
\begin{subequations}\label{max.distr}
\begin{align}
&\nu^*_k(\Theta^0(k))=\sum_{\theta_k\in \Theta^0(k)}\mu_k(\theta_k)+\alpha(k),\label{max.1a}\\
&\nu^*_k(\Theta_0(k))=\big(\sum_{\theta_k\in \Theta_0(k)}\mu_k(\theta_k)-\alpha(k)\big)^+,\label{max.1b}\\
&\nu^*_k(\Theta_s(k))=\big(\sum_{\theta_k\in \Theta_s(k)}\mu_k(\theta_k)\nonumber\\
&\qquad \qquad\ -\big(\alpha(k)-\sum_{j=1}^s\sum_{\theta_k\in \Theta_{j-1}(k)}\mu_k(\theta_k)\big)^+\big)^+,\label{max.1c}\\
& \mbox{for all $s=1,2,\dots,r$, and}\nonumber\\
&\alpha(k)=\min\big(R_{TV}(k),1-\sum_{\theta_k\in \Theta^0(k)}\mu_k(\theta_k)\big).\label{max.1d}
\end{align}
\end{subequations}
\end{theorem}
\begin{proof}
See \cite[Theorem 4.1]{ctlthem:2013}.
\end{proof}

The key feature of \eqref{max.distr} is its explicit characterization as a water-filling solution with respect to the nominal posterior mode probabilities $\mu_k(\theta_k)$ and the TVD parameter $R_{TV}(k)$. 

\subsubsection{Equivalent solution} 
By utilizing the results of Theorem \ref{main.theorem},  an equivalent expression for \eqref{payoff} can be provided. 
\begin{theorem}
The optimal solution to the inner optimization problem of \eqref{DRO1} is equivalently given by
\begin{multline}\label{equiv.expr}
V_{k}(\varphi,\nu^*)= \sum_{\theta_k\in \Theta}\mu_k(\theta_k)L_k^\varphi (\theta_k)+\beta(\alpha(k))\\
+\alpha(k)\big(L^\varphi_{\max}(k)-L_{\widetilde{\Theta}}^\varphi(k)\big),
\end{multline}
where $\alpha(k)$ is defined in \eqref{max.1d}. The  set of modes $\widetilde{\Theta}$ and the function $\beta(\alpha(k))$ are calculated as follows.

\emph{Case.1.} If 
\begin{equation}\label{condition.1}
\alpha(k)\leq \sum_{\theta_k\in \Theta_0(k)}\mu_k(\theta_k),
\end{equation}
then 
$\widetilde{\Theta}=\Theta_0\ \mbox{and}\  \beta(\alpha(k))=0$.

\emph{Case.2.} For $z\in \{1,2,\dots,r\}$ if
 \begin{equation} \label{condition.2}
 \sum_{\theta_k\in \cup_{l=0}^{z-1} \Theta_l(k)}\mu_k(\theta_k) \leq \alpha(k)\leq \sum_{\theta_k\in \cup_{l=0}^z \Theta_l(k)}\mu_k(\theta_k),
 \end{equation} 
then $\widetilde{\Theta}=\Theta_z$ and
\begin{equation*}
\beta(\alpha(k))= \sum_{s=0}^{z-1} \Big(\sum_{\theta_k\in \Theta_s(k)} \big(L^\varphi_{\Theta_z}(k)-L^\varphi_{\Theta_s}(k)\big)\mu_k(\theta_k)\Big).
\end{equation*}
\end{theorem}

\begin{proof}
See \cite[Section IV]{Tzortzis9157997}.
\end{proof}

The equivalent expression \eqref{equiv.expr} provides an intuitive interpretation of $V_k(\varphi,\nu^*)$.  Compared to the classical case, where only the first term appears, the distributionally robust solution includes extra terms that allow the estimator to compensate for the possibility of uncertain scenarios. These additional terms capture the difference between the maximum and minimum values of $L_k^\varphi(\theta_k)$ across mode sets, scaled by  $\alpha(k)$, which specifies the size of the ambiguity set $\mathbb{B}_k(\mu_k)$.

%The maximizing mode probabilities $\nu^*_k(\theta_k)$, given by \eqref{max.distr}, are expressed as a function of the total variation distance parameter $R_{TV}(k)$, that is, on the parameter that controls the size of the ambiguity set $\mathbb{B}_k(\mu_k)$. 

% A natural question  is how to select the TVD parameter  in \eqref{max.1d} to limit the conservatism of the distributionally robust estimator. Suppose  $0\leq R^1_{TV}(k)\leq R^2_{TV}(k)$, and that the maximizer of $V_k(\varphi,\nu)$ lies in $\mathbb{B}_k(\mu_k)$ with radius $R^1_{TV}(k)$. By the definition of the ambiguity set \eqref{ambiguity.class}, the maximizer also belongs to  $\nu^*_k\in \mathbb{B}_k(\mu_k)$ with radius $R^2_{TV}(k)$.
% Thus, to reduce conservatism,   the inner optimization in \eqref{DRO1} should be solved using the smallest possible value of $R_{TV}(k)$. 

A key consideration is the role of the TVD radius $R_{TV}(k)$, which determines the level of conservatism of the estimator. Larger values expand the ambiguity set $\mathbb{B}_k(\mu_k)$ and yield more conservative solutions, while smaller values reduce conservatism. The choice of $R_{TV}(k)$ therefore influences the trade-off between robustness and performance.

%\begin{equation*}
%V_k(\varphi^*,\nu^*)<V_k(\varphi^*,\mu)
%\end{equation*}
%where $V_k(\varphi^*,\nu^*)$ denotes the solution of the distributionally robust state estimation problem, and $V_k(\varphi^*,\mu)$ denotes the solution of the estimation problem without considering any uncertainty on the posterior mode probabilities.

Let $D_k:=V_k(\varphi,\nu^*)$. From equation \eqref{equiv.expr}, under condition \eqref{condition.1}, we obtain
\begin{equation}\label{alpha.cond1}
\alpha(k)=\frac{D_k-\sum_{\theta_k\in \Theta}\mu_k(\theta_k)L_k^\varphi(\theta_k)}{L_{\max}^\varphi(k)-L_{\min}^\varphi(k)}.
\end{equation}
Substituting \eqref{alpha.cond1} back into \eqref{condition.1}, and solving for $D_k$, we get that \eqref{alpha.cond1} holds for
\begin{multline}
D_k\leq \sum_{\theta_k\in \Theta}\mu_k(\theta_k)L_k^\varphi (\theta_k)\\
+\big\{L^\varphi_{\max}(k)-L^\varphi_{\min}(k)\big\}\sum_{\theta_k\in \Theta_0(k)}\mu_k(\theta_k).
\end{multline}
Similarly,  under condition \eqref{condition.2}, equation \eqref{equiv.expr} becomes 
\begin{align}
V_k(\varphi,\nu^*)&=\sum_{\theta_k\in \Theta}\mu_k(\theta_k)L_k^\varphi (\theta_k){+}\alpha(k)\big\{L^\varphi_{\max}(k){-}L^\varphi_{\Theta_z}(k)\big\}\nonumber\\
 &+\sum_{s=0}^{z-1}\sum_{\theta_k\in \Theta_s(k)}\mu_k(\theta_k)\big\{L^\varphi_{\Theta_z}(k){-}L^\varphi_{\Theta_s}(k)\big\}.
\end{align}
It follows that,
\begin{equation}\label{alpha.cond2}
\alpha(k)=\frac{D_k-\Lambda(k)}{L_{\max}^\varphi(k)-L_{\Theta_z}^\varphi(k)},
\end{equation}
where 
\begin{align*}
&\Lambda(k)=L_{\max}^\varphi(k)\sum_{\theta_k\in \Theta^0(k)}\mu_k(\theta_k)\\
& {-}L^\varphi_{\Theta_z}(k)\sum_{s=0}^{z-1}\sum_{\theta_k\in \Theta_s(k)}\mu_k(\theta_k){-}\sum_{s=z}^{r}\sum_{\theta_k\in \Theta_s(k)}\mu_k(\theta_k)L^\varphi_{\Theta_s}(k).
\end{align*}
Substituting \eqref{alpha.cond2} into \eqref{condition.2} and solving for $D_k$, we get that \eqref{alpha.cond2} holds for
\begin{align*}
D_k&\leq L_{\max}^\varphi(k)\big(\sum_{s=0}^{z-1}\sum_{\theta_k\in\Theta_s(k)}\mu_k(\theta_k)+\sum_{\theta_k\in \Theta^0(k)}\mu_k(\theta_k)\big)\\
&\quad +\sum_{s=z-1}^r\sum_{\theta_k\in\Theta_s(k)}\mu_k(\theta_k)L_{\Theta_s}^\varphi(k),\\
D_k&\geq L_{\max}^\varphi(k)\big(\sum_{s=0}^{z-1}\sum_{\theta_k\in\Theta_s(k)}\mu_k(\theta_k)+\sum_{\theta_k\in \Theta^0(k)}\mu_k(\theta_k)\big)\\
&\quad +\sum_{s=z}^r\sum_{\theta_k\in\Theta_s(k)}\mu_k(\theta_k)L_{\Theta_s}^\varphi(k).
\end{align*}
This analysis  illustrates the relationship between the size of the ambiguity set $\mathbb{B}_k(\mu_k)$ specified by $\alpha(k)$, and the solution of the inner optimization in \eqref{DRO1}.

\subsection{Solution to the outer optimization problem}

\subsubsection{Duality}
\begin{lemma} DRO problem \eqref{DRO1} can be equivalently written as
\begin{equation}\label{Equiv_DRO1}
\max_{\nu_k\in \mathbb{B}_k(\mu_k)} \min_{\varphi_k\in\Phi}V_k(\varphi,\nu).
\end{equation}
Problems \eqref{DRO1} and \eqref{Equiv_DRO1} are dual of each other, and the optimal pair $(\varphi_k^*,\nu_k^*)$ is called an equilibrium solution.
\end{lemma}

\begin{proof} 
The result follows from the strong minimax theorem, since $V_k(\varphi,\nu)$ is a continuous, concave--convex function. Specifically, (i) for fixed $\nu$, $V_k(\cdot,\nu)$ is convex because it is quadratic in $\varphi\in\Phi$, and (ii) for fixed $\varphi$, $V_k(\varphi,\cdot)$ is concave because it is linear in $\nu_k\in \mathbb{B}_k(\mu_k)$.
%By the weak min-max property we have that
%\begin{equation*}
%\max_{\nu_k\in \mathbb{B}_k(\mu_k)}\min_{\varphi_k\in\Phi}V_k(\varphi,\nu)\leq \min_{\varphi_k\in \Phi}\max_{\nu_k\in \mathbb{B}_k(\mu_k)} V_k(\varphi,\nu).
%\end{equation*}
%Suppose that the pair $(\varphi_k^*, \nu_k^*)$ solve the max-min problem \eqref{Equiv_DRO1}. Then, $V_k(\varphi^*,\nu^*)\leq \min_{\varphi}\max_{\nu}V_k(\varphi,\nu)$. On the other hand, 
%\begin{equation*}
%\min_{\varphi_k\in\Phi}\max_{\nu_k\in \mathbb{B}_k(\mu_k)}V_k(\varphi,\nu)\leq \max_{\nu_k\in \mathbb{B}_k(\mu_k)}V_k(\varphi^*,\nu).
%\end{equation*}
%Since $\nu^*_k$  maximizes the right-hand side, then
%\begin{equation*}
%\min_{\varphi_k\in\Phi}\max_{\nu_k\in \mathbb{B}_k(\mu_k)}V_k(\varphi,\nu)\leq V_k(\varphi^*,\nu^*).
%\end{equation*}
%Hence, the strong max-min property  holds, that is, \begin{equation*}
%\max_{\nu_k\in \mathbb{B}_k(\mu_k)}\min_{\varphi_k\in \Phi}V_k(\varphi,\nu)= \min_{\varphi_k\in \Phi}\max_{\nu_k\in \mathbb{B}_k(\mu_k)} V_k(\varphi,\nu).
%\end{equation*}
\end{proof}

\subsubsection{Optimal estimator}
The above duality relation allow us to solve the inner optimization in \eqref{Equiv_DRO1} for every $\nu_k\in \mathbb{B}_k(\mu_k)$.

\begin{proposition}
The solution of the inner optimization in \eqref{Equiv_DRO1} is the conditional expectation
\begin{equation}\label{opt.estimator}
\hat{\varphi}(y_{1:k})=\mathbb{E}\big[x_k|y_{1:k}\big]=\int x_kp^{\nu^*}(x_k|y_{1:k})dx_k,
\end{equation}
where $p^{\nu^*}(\cdot)$ indicates the dependence on the maximizing posterior mode probability. The corresponding MMSE is the conditional variance of $x_k$ given the  measurements $y_{1:k}$.
\end{proposition}

\begin{proof}
The result follows directly by noting that the MSE cost \eqref{mse.DRO.cost} can be written as $V_k(\varphi,\nu)=\mathbb{E}\big[||x_k-\varphi(y_{1:k})||^2\big]$
% \begin{equation*}
% V_k(\varphi,\nu)=\mathbb{E}\big[||x_k-\varphi(y_{1:k})||^2\big]
% \end{equation*}
where the expectation is taken with respect to $p^{\nu^*}(x_k|y_{1:k})=\sum_{\theta_k}p(x_k|\theta_k,y_{1:k})\nu^*_k(\theta_k)$. The proof follows the classical MMSE derivation and is  omitted.
%Expanding the mean-square error optimality criterion we get that
%\begin{align*}
%&V(\varphi,p)=\mathbb{E}\Big[\phi^T(x_k,\theta_k)\phi(x_k,\theta_k)-2\phi^T(x_k,\theta_k)\varphi(y_{1:k-1})\\
% &\quad + \varphi^T(y_{1:k-1})\varphi(y_{1:k-1})|y_{1:k-1}\Big ]\\
% &=\mathbb{E}\Big[\phi^T(x_k,\theta_k)\phi(x_k,\theta_k)|y_{1:k-1}\Big]+\varphi^T(y_{1:k-1})\varphi(y_{1:k-1})\\
% &\quad -2\varphi^T(y_{1:k-1})\mathbb{E}\Big[\phi(x_k,\theta_k)|y_{1:k-1}\Big] \\
% &= \mathbb{E}\Big[||\varphi(y_{1:k-1})-\mathbb{E}[\phi(x_k,\theta_k)|y_{1:k-1}]||^2\Big]\\
% &\quad +\mathbb{E}\Big[\phi^T(x_k,\theta_k)\phi(x_k,\theta_k)|y_{1:k-1}\Big]{-}||\mathbb{E}[\phi(x_k,\theta_k)|y_{1:k-1}]||^2.
%\end{align*}
%The minimum is achieved by setting $\varphi(y_{1:k-1})=\hat{\varphi}(y_{1:k-1})=\mathbb{E}[\phi(x_k,\theta_k)|y_{1:k-1}]$. The corresponding minimum value of the mean-square error  is then
%\begin{equation}\label{min_MSE}
%V(\hat{\varphi},\Pi)=\mathbb{E}\Big[||\phi(x_k,\theta_k)||^2|y_{1:k-1}\Big]-||\hat{\varphi}(y_{1:k-1})||^2.
%\end{equation}
\end{proof}

We have derived the optimal solution to the DRO problem. However, obtaining its analytical solution is intractable due to exponential complexity and the need to evaluate indefinite integrals involved in the prediction and correction steps, as well as in the evaluation of \eqref{expected.term}. We therefore propose a suboptimal algorithm under a Gaussian hypothesis. 

%\subsubsection{Estimated mode probability} Let $\phi(x_k,\theta_k)=\theta_k$. Then, the optimal estimator \eqref{opt.estimator} is given by
%\begin{equation}
%\hat{\varphi}(y_{1:k-1})=\sum_{\theta_{k}}\theta_kp(\theta_k|y_{1:k-1}).
%\end{equation}

\subsection{Suboptimal Solution}\label{suboptimal.sol}
For numerical implementation,  we propose a distributionally robust GPB (DRGPB) algorithm, which is executed by the remote estimator.

\begin{algorithm}[t]
\caption{The  DRGPB algorithm}\label{GPB1.algorithm}
 Input.   $\hat{x}^i_{0}$, $P_0^i$, $\mu_0(\theta_0=i)$, $y_k$, and $R_{TV}(\theta_k=i)\in[0,1]$.
\noindent For $k=1,2,\dots,N$ do:\\
 \noindent \textbf{Step.1}
 \begin{enumerate}
 \item[a:] Evaluate \eqref{KF.equations} to get $\hat{x}_k^i$ and $P_k^i$ for all $i\in \Theta$.
 \item[b:] Evaluate $\mu_k$ using \eqref{posterior.mode.prob}.
 \end{enumerate}
 \noindent \textbf{Step.2}
  \begin{enumerate}
 \item[a:] Evaluate $\Theta^0(k)$, $\Theta_0(k)$, and $\Theta_j(k)$ ($j=1,\dots,r$) using \eqref{avoid.int}.
 \item[b:] Evaluate $\nu^*_k$ as in \eqref{max.distr}.
 \end{enumerate}
  \noindent \textbf{Step.3}
  Evaluate the overall estimates $\hat{x}_k$ and $P_k$ as in \eqref{overall.x} and \eqref{overall.P}, respectively.
  \end{algorithm}

The DRGPB algorithm runs parallel Kalman filters, one for each mode, to compute the conditional state estimates $\hat{x}^j_{k}$ and covariances $P^j_{k}$, for all $j\in \Theta$. These estimates are then merged using the worst-case posterior mode probabilities to obtain the overall estimates $\hat{x}_{k}$ and  $P_{k}$.
For each mode $\theta_k=j\in \Theta$, we define the mode-conditioned state estimates and covariances as
\begin{align*}
\hat{x}^j_{k|k-1}&:=\mathbb{E}[x_k^j|y_{1:k-1}],\qquad \hat{x}^j_{k}:=\mathbb{E}[x_k^j|y_{1:k}],\\
P^j_{k|k-1}&:=\mathbb{E}[(x^j_k-\hat{x}^j_{k|k-1})(x^j_k-\hat{x}^j_{k|k-1})^T],\\
P^j_{k}&:=\mathbb{E}[(x^j_k-\hat{x}^j_{k})(x^j_k-\hat{x}^j_{k})^T].
\end{align*}
Then, the mode-matched Kalman filter prediction-correction equations can be summarized as follows.
\begin{subequations}\label{KF.equations}
\begin{align}
\hat{x}^j_{k|k-1}&=A(\theta_k=j)\hat{x}^j_{k-1},\\
P^j_{k|k-1}&= A(\theta_k=j)P_{k-1} A^T(\theta_k=j)\nonumber\\
& \quad +B(\theta_k=j)W_{k} B^T(\theta_k=j),\\
S_k^j&=C(\theta_k=j)P^j_{k|k-1}C^T(\theta_k=j)\nonumber\\
&\quad +D(\theta_k=j)V_kD^T(\theta_k=j),\\
K_k^j&=P_{k|k-1}^jC^T(\theta_k=j)(S_k^j)^{-1},\\
\hat{x}_k^j&=\hat{x}_{k|k-1}^j+K_k^j(y_k-C(\theta_k=j)\hat{x}^j_{k|k-1}),\\
P_k^j&=P^j_{k|k-1}-K^j_kS_k^j(K_k^j)^T.
\end{align}
\end{subequations}
For each mode $j\in \Theta$, the posterior nominal mode probability is given by
\begin{equation}\label{posterior.mode.prob}
\mu_k(\theta_k=j)=\frac{{\cal N}(\tilde{y}^j_k;0,S_k^j)\sum_{i=1}^{n_\theta}\pi_{ij}\mu_{k-1}(\theta_{k-1}{=}i)}{\sum_{j=1}^{n_\theta}{\cal N}(\tilde{y}^j_k;0,S_k^j)\sum_{i=1}^{n_\theta}\pi_{ij}\mu_{k-1}(\theta_{k-1}{=}i)},
\end{equation}
where $\tilde{y}^j_k=y_k-C(\theta_k=j)\hat{x}^j_{k|k-1}$.

In the classical GPB algorithm, the overall estimate $\hat{x}_k$ and covariance $P_k$ are obtained by 
\begin{align*}
\hat{x}_k&=\sum_{j=1}^{n_\theta}\mu_k(\theta_k=j)\hat{x}_k^j,\\
P_k&=\sum_{j=1}^{n_\theta}\mu_k(\theta_k=j)[P^j_{k}+(\hat{x}_{k}-\hat{x}_{k}^j)(\hat{x}_{k}-\hat{x}_{k}^j)^T].
\end{align*}
This corresponds to assuming that the computed $\mu_k$ coincides with the true posterior distribution.
In the proposed DRGPB algorithm, the same structure is preserved, but the overall estimate $\hat{x}_k$ and covariance $P_k$ are computed using the worst-case posterior mode probabilities $\nu^*_k$. By Theorem \ref{main.theorem}, we have that
\begin{equation*}
V_k(\varphi,\nu^*)=\sum_{i=1}^{n_\theta} \nu^*_k(\theta_k=i)L_k^\varphi (\theta_k=i)=\sum_{i=1}^{n_\theta}\mbox{Tr}(P_k^i).
\end{equation*}
When $R_{TV}=0$,  $\nu^*_k(\theta_k=j)=\mu_k(\theta_k=j)$. Therefore, the mode sets $\Theta^0(k)$, $\Theta_0(k)$, and $\Theta_j(k)$ ($j=1,\dots,r$) can be evaluated using
\begin{equation}\label{avoid.int}
L_k^\varphi (\theta_k=j)=\frac{\mbox{Tr}(P_k^j)}{\mu_k(\theta_k=j)},\quad \mu_k(\theta_k=j)>0,
\end{equation}
which avoids the  integration in \eqref{expected.term} over $R^{n_x}$. The resulting overall estimates  are 
\begin{align}
\hat{x}_{k}&=\sum_{j\in \Theta^0(k)}\nu^*_k(\theta_k=j)\hat{x}_k^j+\sum_{j\in \Theta_0(k)}\nu^*_k(\theta_k=j)\hat{x}_k^j\nonumber\\
&\quad +\sum_{s=1}^r\sum_{j\in \Theta_s(k)}\nu^*_k(\theta_k=j)\hat{x}_k^j,\label{overall.x}\\
P_{k}&=\sum_{j\in \Theta^0(k)}\nu^*_k(\theta_k=j)\bar{P}_k^j+\sum_{j\in \Theta_0(k)}\nu^*_k(\theta_k=j)\bar{P}_k^j\nonumber\\
&\quad +\sum_{s=1}^r\sum_{j\in \Theta_s(k)}\nu^*_k(\theta_k=j)\bar{P}_k^j,\label{overall.P}
\end{align}
where $\bar{P}_k^j=P^j_{k}+(\hat{x}_{k}-\hat{x}_{k}^j)(\hat{x}_{k}-\hat{x}_{k}^j)^T$.
The DRGPB algorithm is summarized in Algorithm \ref{GPB1.algorithm}.
Note that, to employ Algorithm \ref{GPB1.algorithm}, the TVD parameter $R_{TV}(\theta_k)$ must be specified for each $\theta_k\in \Theta$.

\section{Example}
\label{sec:sim.results}
Consider a MJLS with two states $n_x = 2$ and two  modes $n_\theta=2$. 
The system parameters for each mode are:
\begin{itemize}
    \item Mode 1:
\end{itemize}
\begin{equation*}
A_1{=} \begin{pmatrix}
                                0.99 & 0.05 \\
                                -0.10 & 0.95 
\end{pmatrix}, \ 
B_1{=} I, \ C_1{=} \begin{pmatrix}
                                1 & 0.5 \\
                                1 & 1
\end{pmatrix},  \
D_1{=}I.
\end{equation*}
\begin{itemize}
    \item Mode 2:
\end{itemize}
\begin{equation*}
A_2{=} \begin{pmatrix}
                                0.65 & 0.09 \\
                                -0.35 & 0.10
\end{pmatrix},  \
B_2{=} I,\
C_2{=} \begin{pmatrix}
                                1 & 0.5 \\
                                0.5 & 1
\end{pmatrix},  \ 
D_2{=} I.
\end{equation*}
The initial state and noise processes are  $x_0\sim {\cal N}(x_0;0,I)$, $w_k\sim {\cal N}(w_k;0,I)$ and  $v_k\sim {\cal N}(v_k;0,I)$, where $I=\mbox{diag}(1,1)$. The MC initial distribution is given by $[0.4\ 0.6]^T$.

In this example, the system is monitored by a remote estimator that has access only to the observation sequence $y_{1:k}$. Although the simulations do not explicitly model packet drops or channel noise, the estimator computes both the state estimate and the posterior mode probabilities using only the received measurements.
The uncertainty in the posterior mode probabilities is assumed to arise  from uncertainty in the MC transition probabilities.  Two  scenarios are considered:
\begin{itemize}
\item \emph{Scenario I - Low uncertainty:}
For  $1 \leq k< 30$, the true and nominal transition probability matrices are 
\begin{equation*}
\Pi^{\text{true}}(k){=} \begin{pmatrix}
                                0.65 & 0.35 \\
                                0.40 & 0.60 
\end{pmatrix},\
\Pi^{\text{nom}}(k){=} \begin{pmatrix}
                                0.60 & 0.40 \\
                                0.45 & 0.55 
\end{pmatrix}.
\end{equation*}
\item \emph{Scenario II - High uncertainty:} For $70 \leq k\leq N$, the true and nominal  transition probability matrices are 
\begin{equation*}
\Pi^{\text{true}}(k){=} \begin{pmatrix}
                                0.15 & 0.85 \\
                                0.05 & 0.95
\end{pmatrix},\
\Pi^{\text{nom}}(k){=} \begin{pmatrix}
                                0.95 & 0.05 \\
                                0.90 & 0.10 
\end{pmatrix}.
\end{equation*}
\end{itemize}

We compare the solutions obtained from the classical state estimation problem and the distributionally robust state estimation problem. Fig. \ref{fig1}  depicts the estimation results  from a single realization of the hybrid state and observation process for $N=100$. In the simulations, DRGPB Algorithm \ref{GPB1.algorithm} was run using the nominal MC transition probability matrix $\Pi^{nom}(k)$, while the modes transitions $\theta_k$ were generated using the true  transition matrix $\Pi^{true}(k)$. 

Fig. \ref{fig1}\subref{fig_state_nom1}-\subref{fig_state_nom2} depict the hidden state $x_k$ and the corresponding state estimate $\hat{x}_k$ for $R_{TV}(k)=0$, $\forall k$. In this case, $\nu_k^*(\theta_k)=\mu_k(\theta_k)$, and the distributionally robust state estimation problem \eqref{DRO1} reduces to the classical formulation without uncertainty. Consequently, the DRGPB Algorithm \ref{GPB1.algorithm} is equivalent to the standard GPB algorithm.  
The impact of uncertainty on the estimator's performance  is clearly visible in Fig. \ref{fig1}\subref{fig_error}.
In contrast, Fig. \ref{fig1}\subref{fig_state_true1}-\subref{fig_state_true2} depict $x_k$ and  $\hat{x}_k$ when Algorithm \ref{GPB1.algorithm} is applied. The improved  estimation performance is evident in Fig. \ref{fig1}\subref{fig_error}.

The true Markov state is plotted in Fig. \ref{fig1}\subref{fig_prob} (top plot). Both the nominal and maximizing posterior mode probabilities are plotted in Fig. \ref{fig1}\subref{fig_prob}  for the two operating modes (middle and bottom plots). For $1\leq k< 30$ (low uncertainty scenario I), the maximizing posterior mode probability $\nu^*_k(\theta_k)$ closely follows the nominal probabilities $\mu_k(\theta_k)$. In contrast, during the high uncertainty scenario II ($70 \leq k\leq N$),  $\nu^*_k(\theta_k)$ reflects the worst-case posterior distribution within the ambiguity set, enabling more accurate identification of the true mode $\theta_k$.

\begin{figure*}[htbp]
\centering
%----start of first subfigure----
\subfloat[][]{
\label{fig_state_nom1} %% label for first subfigure
\includegraphics[ width=.29\linewidth]{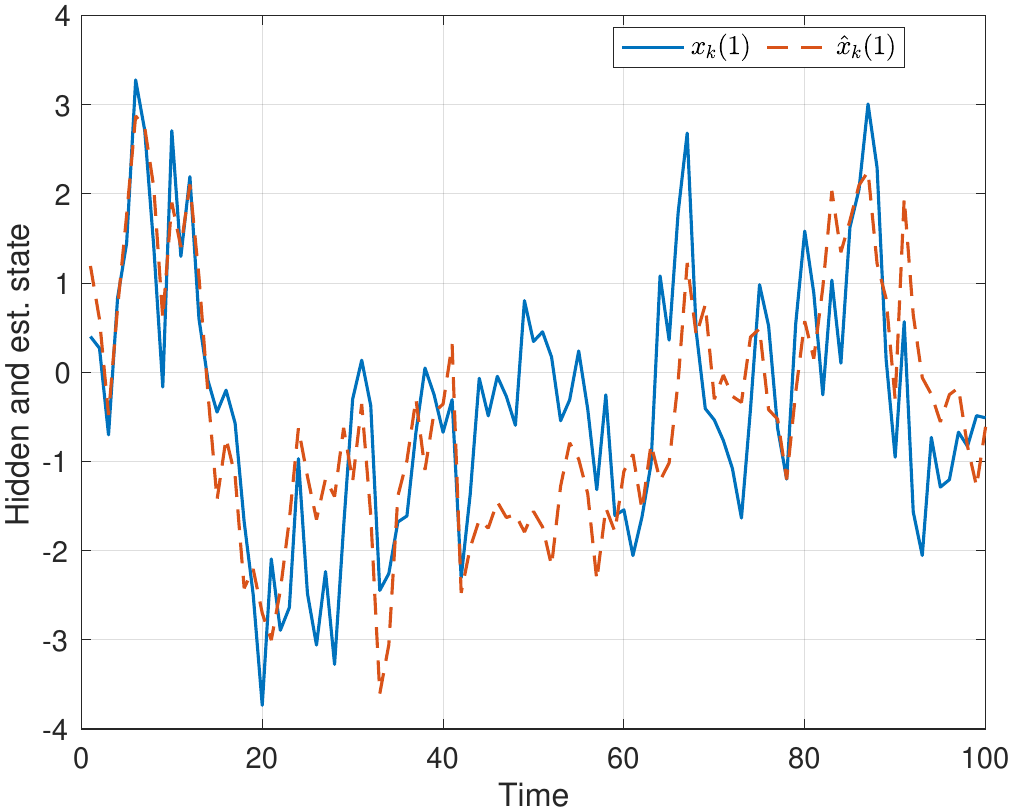}}
\subfloat[][]{
\label{fig_state_nom2} %% label for first subfigure
\includegraphics[ width=.29\linewidth]{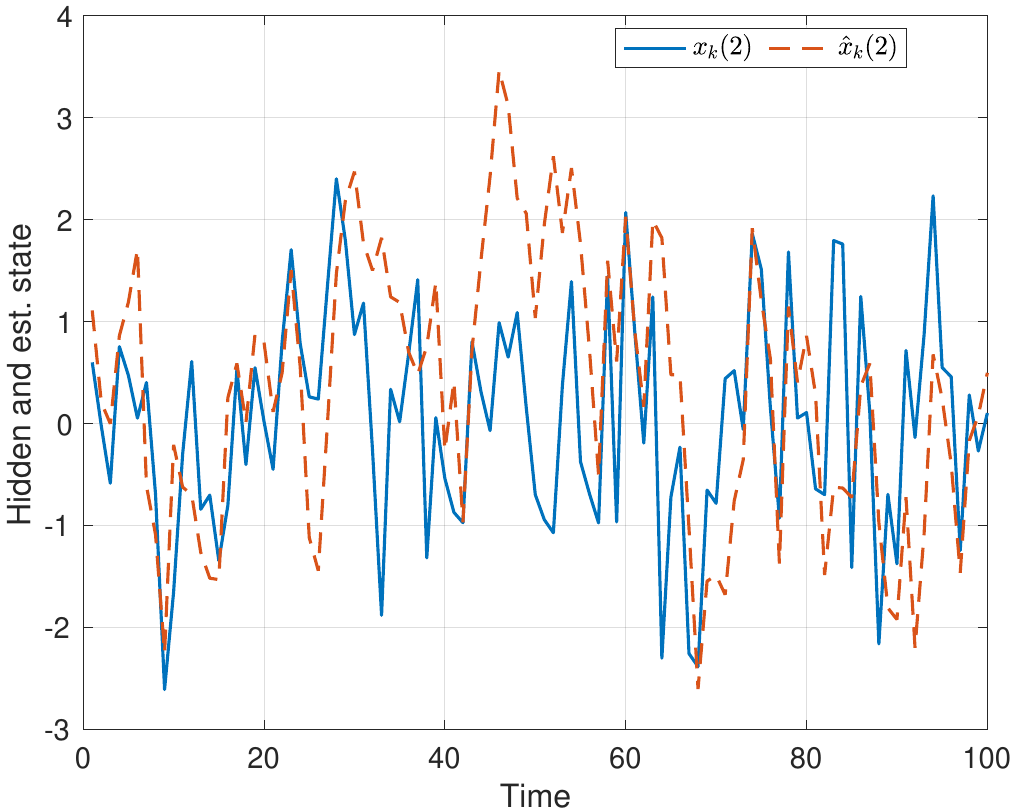}}
%----start of second subfigure----
\subfloat[][]{
\label{fig_prob} %% label for second subfigure
\includegraphics[ width=.29\linewidth]{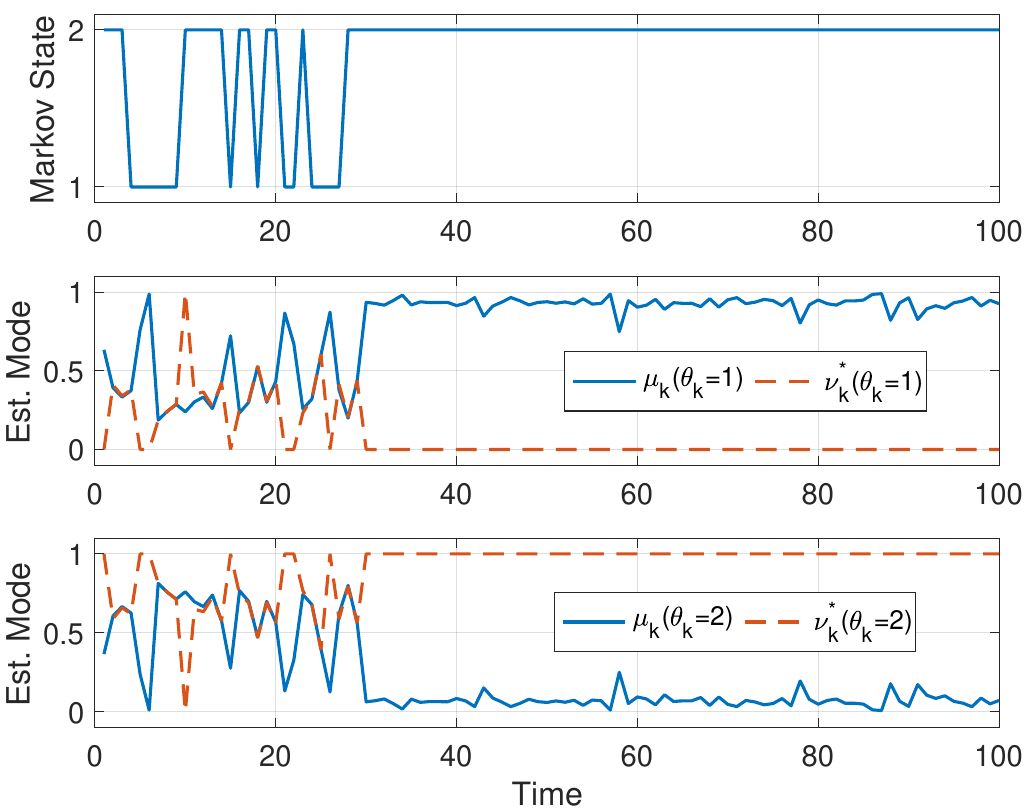}}\\
\subfloat[][]{
\label{fig_state_true1} %% label for second subfigure
\includegraphics[width=.29\linewidth]{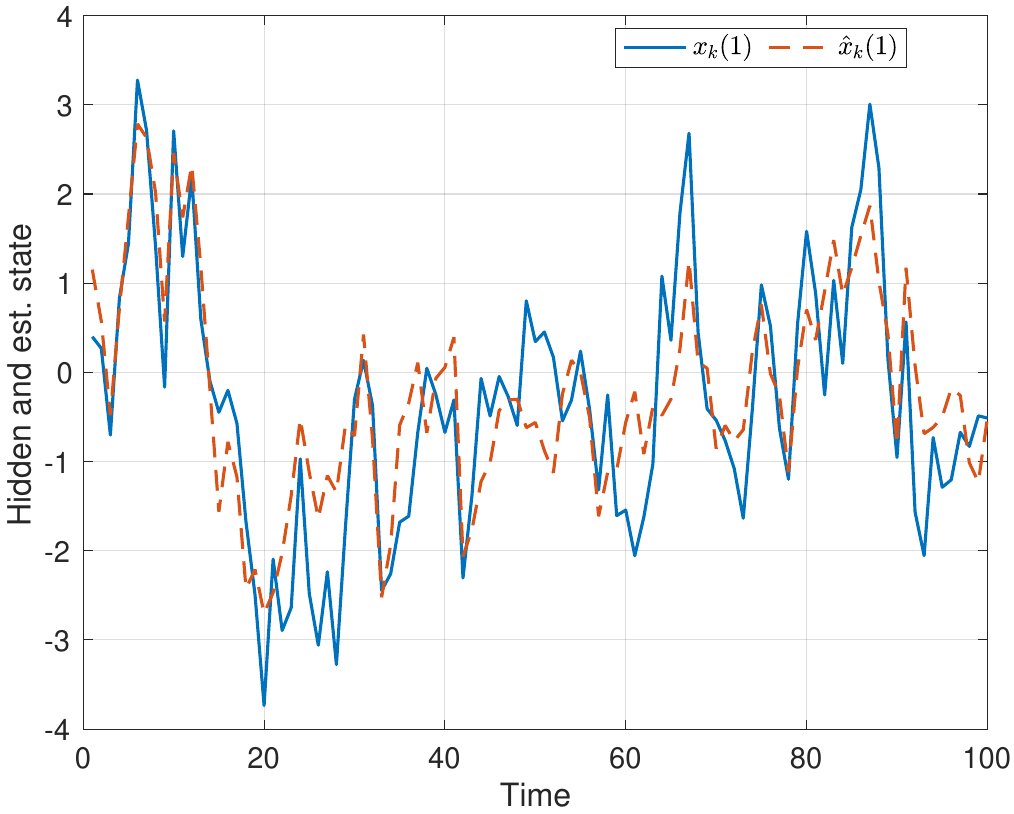}}
\subfloat[][]{
\label{fig_state_true2} %% label for second subfigure
\includegraphics[width=.29\linewidth]{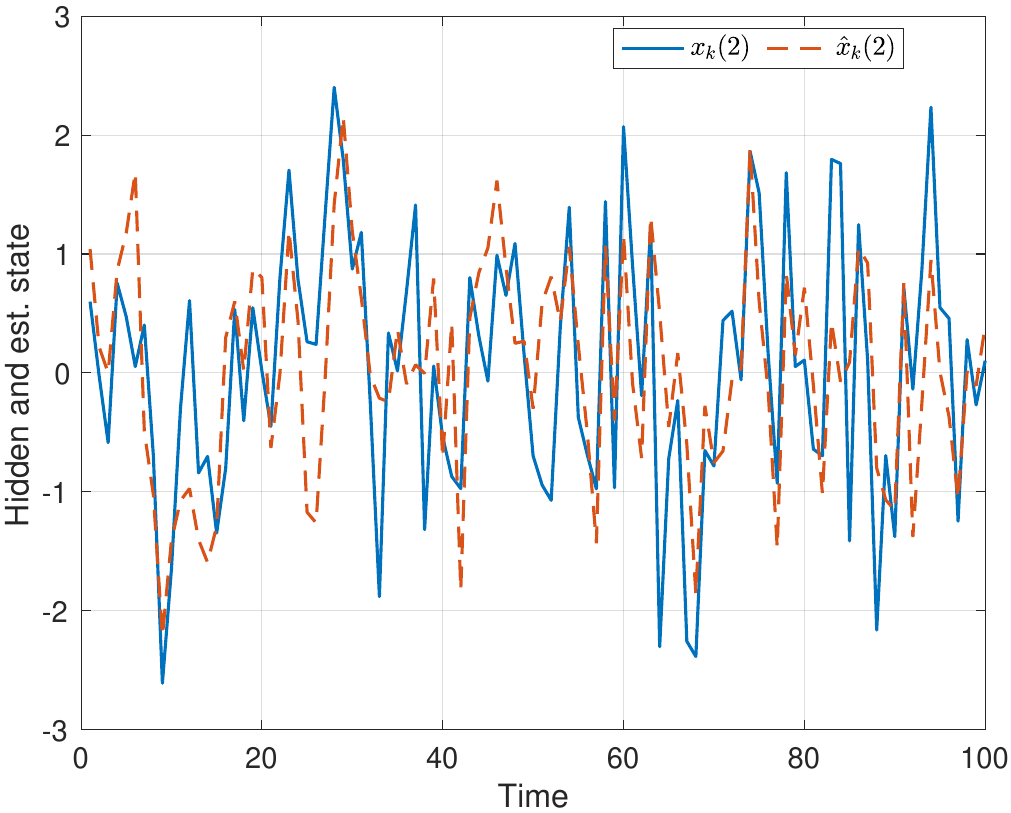}}
\subfloat[][]{
\label{fig_error} %% label for second subfigure
\includegraphics[ width=.29\linewidth]{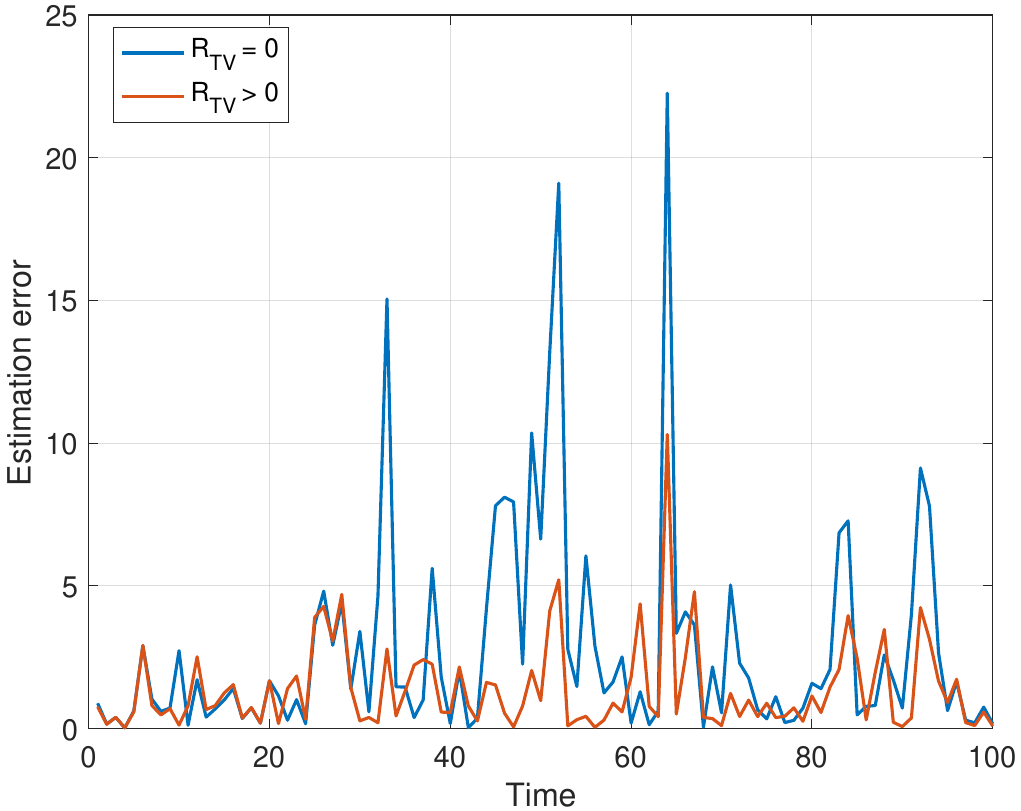}}
\caption{Simulation results comparing the classical and distributionally robust state estimation. Plots (a)--(b) and (d)--(e) depict the hidden state $x_k$ and the state estimate $\hat{x}_k$ for $R_{TV}=0$ and $R_{TV}>0$, respectively.  Plot (c) depicts the true Markov state $\theta_k$ (top), along with the nominal posterior mode probabilities $\mu_k(\theta_k)$ and the maximizing posterior mode probabilities $\nu_k^*(\theta_k)$ for the two operating modes (middle and bottom). Plot (f) depicts the estimation error for $R_{TV}=0$ and $R_{TV}>0$.}
\label{fig1}
\end{figure*}

\section{Conclusions}
\label{sec:conclusions}
We proposed a distributionally robust approach for remote estimation of MJLS under uncertainty in the posterior mode probabilities. TVD metric was used to quantify uncertainty by defining ambiguity sets centered at the estimator’s nominal posterior.
The resulting  estimator extends the classical solution with correction terms that  mitigate the impact of incorrect posterior mode probabilities. For numerical implementation, we introduced a distributionally robust variant of the widely used  GPB algorithm.
Simulation results confirmed that the proposed approach limits the influence of uncertainty and preserves estimation performance, even under significant mismatch in the mode probabilities.

\bibliographystyle{IEEEtran}
\bibliography{autosam}

\begin{thebibliography}{10}
\providecommand{\url}[1]{#1}
\csname url@samestyle\endcsname
\providecommand{\newblock}{\relax}
\providecommand{\bibinfo}[2]{#2}
\providecommand{\BIBentrySTDinterwordspacing}{\spaceskip=0pt\relax}
\providecommand{\BIBentryALTinterwordstretchfactor}{4}
\providecommand{\BIBentryALTinterwordspacing}{\spaceskip=\fontdimen2\font plus
\BIBentryALTinterwordstretchfactor\fontdimen3\font minus
  \fontdimen4\font\relax}
\providecommand{\BIBforeignlanguage}[2]{{%
\expandafter\ifx\csname l@#1\endcsname\relax
\typeout{** WARNING: IEEEtran.bst: No hyphenation pattern has been}%
\typeout{** loaded for the language `#1'. Using the pattern for}%
\typeout{** the default language instead.}%
\else
\language=\csname l@#1\endcsname
\fi
#2}}
\providecommand{\BIBdecl}{\relax}
\BIBdecl

\bibitem{logothetis1999expectation}
A.~Logothetis and V.~Krishnamurthy, ``Expectation maximization algorithms for
  {MAP} estimation of jump {M}arkov linear systems,'' \emph{IEEE Trans. Signal
  Process.}, vol.~47, no.~8, pp. 2139--2156, 1999.

\bibitem{mcginnity2000multiple}
S.~McGinnity and G.~W. Irwin, ``Multiple model bootstrap filter for maneuvering
  target tracking,'' \emph{IEEE Trans. Aerosp. Electron. Syst.}, vol.~36,
  no.~3, pp. 1006--1012, 2000.

\bibitem{magnant2016bayesian}
C.~Magnant, A.~Giremus, E.~Grivel, L.~Ratton, and B.~Joseph, ``Bayesian
  non-parametric methods for dynamic state-noise covariance matrix estimation:
  Application to target tracking,'' \emph{Signal Processing}, vol. 127, pp.
  135--150, 2016.

\bibitem{tzortzis2021linear}
I.~Tzortzis, C.~N. Hadjicostis, and C.~D. Charalambous, ``Linear quadratic
  tracking control of hidden {M}arkov jump linear systems subject to
  ambiguity,'' in \emph{IEEE Conf. Decis. and Control}, 2021, pp. 2336--2341.

\bibitem{liu2017state}
W.~Liu, ``State estimation for discrete-time {M}arkov jump linear systems with
  time-correlated measurement noise,'' \emph{Automatica}, vol.~76, pp.
  266--276, 2017.

\bibitem{yin2016distributed}
X.~Yin, Z.~Li, L.~Zhang, and M.~Han, ``Distributed state estimation of
  sensor-network systems subject to {M}arkovian channel switching with
  application to a chemical process,'' \emph{IEEE Trans. on Syst. Man Cybern.
  Syst.}, vol.~48, no.~6, pp. 864--874, 2016.

\bibitem{zhao2016detection}
S.~Zhao, B.~Huang, and F.~Liu, ``Detection and diagnosis of multiple faults
  with uncertain modeling parameters,'' \emph{IEEE Trans. Control Syst.
  Technol.}, vol.~25, no.~5, pp. 1873--1881, 2016.

\bibitem{hibey1999conditional}
J.~L. Hibey and C.~D. Charalambous, ``Conditional densities for continuous-time
  nonlinear hybrid systems with applications to fault detection,'' \emph{IEEE
  Trans. Autom. Control}, vol.~44, no.~11, pp. 2164--2169, 1999.

\bibitem{ahn2019bayesian}
S.~Zhao, C.~K. Ahn, P.~Shi, Y.~Shmaliy, and F.~Liu, ``Bayesian state estimation
  for {M}arkovian jump systems: Employing recursive steps and pseudocodes,''
  \emph{IEEE Syst. Man Cybern. Mag.}, vol.~5, no.~2, pp. 27--36, 2019.

\bibitem{blom1988interacting}
H.~A. Blom and Y.~Bar-Shalom, ``The interacting multiple model algorithm for
  systems with {M}arkovian switching coefficients,'' \emph{IEEE Trans. Autom.
  Control}, vol.~33, no.~8, pp. 780--783, 1988.

\bibitem{elliott2005state}
R.~J. Elliott, F.~Dufour, and W.~P. Malcolm, ``State and mode estimation for
  discrete-time jump {M}arkov systems,'' \emph{SIAM J. Control Optim.},
  vol.~44, no.~3, pp. 1081--1104, 2005.

\bibitem{Tzortzis:2025}
I.~Tzortzis, E.~Makridis, C.~D. Charalambous, and T.~Charalambous, ``Remote
  estimation over packet-dropping wireless channels with partial state
  information,'' in \emph{Proc. of the European Control Conf.}, 2025, pp.
  1414--1419.

\bibitem{schenato2007foundations}
L.~Schenato, B.~Sinopoli, M.~Franceschetti, K.~Poolla, and S.~S. Sastry,
  ``Foundations of control and estimation over lossy networks,''
  \emph{Proceedings of the IEEE}, vol.~95, no.~1, pp. 163--187, 2007.

\bibitem{wang2023distributionally}
S.~Wang, ``Distributionally robust state estimation for jump linear systems,''
  \emph{IEEE Trans. Signal Process.}, vol.~71, pp. 3835--3851, 2023.

\bibitem{jilkov2004online}
V.~P. Jilkov and X.~R. Li, ``Online {B}ayesian estimation of transition
  probabilities for {M}arkovian jump systems,'' \emph{IEEE Trans. Signal
  Process.}, vol.~52, no.~6, pp. 1620--1630, 2004.

\bibitem{doucet2002recursive}
A.~Doucet and B.~Ristic, ``Recursive state estimation for multiple switching
  models with unknown transition probabilities,'' \emph{IEEE Trans. Aerosp.
  Electron. Syst.}, vol.~38, no.~3, pp. 1098--1104, 2002.

\bibitem{li2016h}
X.~Li, J.~Lam, H.~Gao, and J.~Xiong, ``${H}_{\infty}$ and ${H}_2$ filtering for
  linear systems with uncertain {M}arkov transitions,'' \emph{Automatica},
  vol.~67, pp. 252--266, 2016.

\bibitem{orguner2006online}
U.~Orguner and M.~Demirekler, ``An online sequential algorithm for the
  estimation of transition probabilities for jump {M}arkov linear systems,''
  \emph{Automatica}, vol.~42, no.~10, pp. 1735--1744, 2006.

\bibitem{zhang2009mode}
L.~Zhang and E.-K. Boukas, ``Mode-dependent ${H}_{\infty}$ filtering for
  discrete-time {M}arkovian jump linear systems with partly unknown transition
  probabilities,'' \emph{Automatica}, vol.~45, no.~6, pp. 1462--1467, 2009.

\bibitem{ctlthem:2013}
C.~Charalambous, I.~Tzortzis, S.~Loyka, and T.~Charalambous, ``Extremum
  problems with total variation distance and their applications,'' \emph{IEEE
  Trans. Autom. Control}, vol.~59, no.~9, pp. 2353--2368, 2014.

\bibitem{Tzortzis9157997}
I.~Tzortzis, C.~D. Charalambous, and C.~N. Hadjicostis, ``Jump {LQR} systems
  with unknown transition probabilities,'' \emph{IEEE Trans. Autom. Control},
  vol.~66, no.~6, pp. 2693--2708, 2021.

\end{thebibliography}

%\flushend

\end{document}